%% file: finitestatedelay.tex
\documentclass[submission,copyright,creativecommons]{eptcs}
\providecommand{\event}{GandALF 2017} 

\input{preamble}

\title{Finite-state Strategies in Delay Games\thanks{Supported by the project \myquot{TriCS} (ZI 1516/1-1) of the German Research Foundation (DFG).}}
\author{Martin Zimmermann
\institute{Reactive Systems Group, Saarland University, 66123 Saarbrücken, Germany}
\email{zimmermann@react.uni-saarland.de}
}
\def\titlerunning{Finite-state Strategies in Delay Games}
\def\authorrunning{M.\ Zimmermann}

\begin{document}
\maketitle

\begin{abstract}
\input{content/abstract}
\end{abstract}

\section{Introduction}
\label{sec-intro}
\input{content/intro}

\section{Preliminaries}
\label{sec-prel}
\input{content/prelims}

\section{What is a Finite-state Strategy in a Delay Game?}
\label{sec-fsindg}
\input{content/finitestateindelaygame}

\section{Computing Finite-state Strategies for Block Games}
\label{sec-construction}
\input{content/construction}

\section{Discussion}
\label{sec-disc}
\input{content/disc}

\section{Conclusion}
\label{sec-conc}
\input{content/conc}

\bibliographystyle{eptcs}
\bibliography{biblio}


\end{document}

%% file: preamble.tex
\usepackage{amsmath}
\usepackage{amsthm}
\usepackage{amssymb}
\usepackage{mathtools}
\usepackage{nicefrac}

\theoremstyle{plain}
\newtheorem{theorem}{Theorem}
\newtheorem{lemma}{Lemma}
\newtheorem{example}{Example}
\newtheorem{corollary}{Corollary}

\usepackage[utf8]{inputenc}
\usepackage[T1]{fontenc}

\usepackage{tikz}
\usetikzlibrary{automata}
\usetikzlibrary{backgrounds}

\usepackage{booktabs}

\tikzset{
  initial text=,
  every path/.style={->,-stealth, every loop/.style={-stealth}},
  every initial by arrow/.style={-stealth, initial text={},
  every loop/.style={red,-stealth}},
}

\newcommand{\myquot}[1]{``#1''}

\theoremstyle{plain}
\newtheorem{numberedremark}[theorem]{Remark}

\newcommand{\nats}{\mathbb{N}}
\newcommand{\size}[1]{|#1|}

\renewcommand{\epsilon}{\varepsilon}
\renewcommand{\phi}{\varphi}

\newcommand{\set}[1]{\{#1\}}
\newcommand{\pow}[1]{2^{#1}}

\newcommand{\aut}{\mathfrak{A}}
\newcommand{\autb}{\mathfrak{B}}
\newcommand{\strataut}{\mathfrak{T}}

\newcommand{\acc}{\mathrm{Acc}}
\newcommand{\col}{\Omega}
\newcommand{\infi}[0]{\mathrm{Inf}}
\newcommand{\curlyF}[0]{\mathcal{F}}

\newcommand{\parity}{\mathrm{prty}}
\newcommand{\muller}{\mathrm{mllr}}
\newcommand{\initmark}{I}

\newcommand{\delaygame}[1]{\Gamma\!_{f}(#1)}

\newcommand{\blockgame}[2]{\Gamma^{#1}(#2)}

\newcommand{\SigmaI}{\Sigma_I}
\newcommand{\SigmaO}{\Sigma_O}
\newcommand{\strat}{\tau}
\newcommand{\stratO}{\tau_O}
\newcommand{\stratI}{\tau_I}
\newcommand{\p}{P}

\newcommand{\game}{\Gamma}

\newcommand{\bigo}{\mathcal{O}}

\newcommand{\update}{\mathrm{upd}}

\newcommand{\block}[1]{\overline{#1}}

\newcommand{\equivclass}[1]{[#1]_\equiv}
\newcommand{\R}{{R}}

\newcommand{\mon}{\mathfrak{M}}
\newcommand{\Leq}{L_=}

%% file: content/abstract.tex
What is a finite-state strategy in a delay game? We answer this surprisingly non-trivial question and present a very general framework for computing such strategies: they exist for all winning conditions that are recognized by automata with acceptance conditions that satisfy a certain aggregation property. Our framework also yields upper bounds on the complexity of determining the winner of such delay games and upper bounds on the necessary lookahead to win the game. In particular, we cover all previous results of that kind as special cases of our uniform approach. 

%% file: content/intro.tex
What is a finite-state strategy in a delay game? The answer to this question is surprisingly non-trivial due to the nature of delay games in which one player is granted a lookahead on her opponent's moves. This puts her into an advantage when it comes to winning games, i.e., there are games that can only be won with lookahead, but not without. A simple example is a game where one has to predict the third move of the opponent with one's first move. This is impossible when moving in alternation, but possible if one has access to the opponent's first three moves before making the first move. More intriguingly, lookahead also allows Player~$O$ to improve the quality of her winning strategies in games with quantitative winning conditions, i.e., there is a tradeoff between quality and amount of lookahead~\cite{Zimmermann17}. 

However, managing (and, if necessary, storing) the additional information gained by the lookahead can be a burden. Consider another game where one just has to copy the opponent's moves. This is obviously possible with or without lookahead (assuming the opponent moves first). In particular, without lookahead one just has to remember the last move of the opponent and copy it. However, when granted lookahead, one has to store the last moves of the opponent in a queue to implement the copying properly. This example shows that lookahead is not necessarily advantageous when it comes to minimizing the memory requirements of a strategy.

In this work, we are concerned with Gale-Stewart games~\cite{GaleStewart53}, abstract games without an underlying arena.\footnote{The models of Gale-Stewart games and arena-based games are interreducible, but delay games are naturally presented as a generalization of Gale-Stewart games. This is the reason we prefer this model here.} In such a game, both players produce an infinite sequence of letters and the winner is determined by the combination of these sequences. If it is in the winning condition, a set of such combinations, then the second player wins, otherwise the first one wins. In a classical Gale-Stewart game, both players move in alternation while in a delay game, the second player skips moves to obtain a lookahead on the opponent's moves. Which moves are skipped is part of the rules of the game and known to both players.

Delay games have recently received a considerable amount of attention after being introduced by Hosch and Landweber~\cite{HoschLandweber72} only three years after the seminal Büchi-Landweber theorem~\cite{BuechiLandweber69}. Büchi and Landweber had shown how to solve infinite two-player games with $\omega$-regular winning conditions. Forty years later, delay games were revisited by Holtmann, Kaiser, and Thomas~\cite{HoltmannKaiserThomas12} and the first comprehensive study was initiated, which settled many basic problems like the exact complexity of solving $\omega$-regular delay games and the amount of  lookahead necessary to win such games~\cite{KleinZimmermann16}. Furthermore, Martin's seminal Borel determinacy theorem~\cite{Martin75} for Gale-Stewart games has been lifted to delay games~\cite{KleinZimmermann15} and winning conditions beyond the $\omega$-regular ones have been investigated~\cite{FridmanLoedingZimmermann11,KleinZimmermann16a,Zimmermann16,Zimmermann17}.  Finally, the uniformization problem for relations over infinite words boils down to solving delay games: a relation~$L \subseteq (\SigmaI \times \SigmaO)^\omega$ is uniformized by a continuous function (in the Cantor topology) if, and only if, the delaying player wins the delay game with winning condition~$L$. We refer to~\cite{HoltmannKaiserThomas12} for details.

What makes finite-state strategies in infinite games particularly useful and desirable is that a general strategy is an infinite object, as it maps finite play prefixes to next moves. On the other hand, a finite-state strategy is implemented by a transducer, an automaton with output, and therefore finitely represented: the automaton reads a play prefix and outputs the next move to be taken. Thus, the transducer computes a finite abstraction of the play's history using its state space as memory and determines the next move based on the current memory state.

 In Gale-Stewart games, finite-state strategies suffice for all $\omega$-regular games~\cite{BuechiLandweber69} and even for deterministic $\omega$-contextfree games, if one allows pushdown transducers~\cite{Walukiewicz01}. For Gale-Stewart games (and arena-based games), the notion is well-established and one of the most basic questions about a class of winning conditions is that about the existence and size of winning strategies for such games.

While foundational questions for delay games have been answered and many results have been lifted from Gale-Stewart games to those with delay, the issue of computing tractable and implementable strategies has not been addressed before. However, this problem is of great importance, as the existence and computability of finite-state strategies is a major reason for the successful application of infinite games to diverse problems like reactive synthesis, model-checking of fixed-point logics, and automata theory. 

In previous work, restricted classes of strategies for delay games have been considered~\cite{KleinZimmermann15}. However, those restrictions are concerned with the amount of information about the lookahead's evolution a strategy has access to, and do not restrict the size of the strategies: In general, they are still infinite objects. On the other hand, it is known that bounded lookahead suffices for many winning conditions of importance, e.g., the $\omega$-regular ones~\cite{KleinZimmermann16}, those recognized by parity and Streett automata with costs~\cite{Zimmermann17}, and those definable in (parameterized) linear temporal logics~\cite{KleinZimmermann16a}. Furthermore, for all those winning conditions, the winner of a delay game can be determined effectively. In fact, all these proofs rely on the same basic construction that was already present in the work of Holtmann, Kaiser, and Thomas~\cite{HoltmannKaiserThomas12}, i.e., a reduction to a Gale-Stewart game using equivalence relations that capture the behavior of the automaton recognizing the winning condition. These reductions and the fact that finite-state strategies suffice for the games obtained in the reductions imply that (some kind of) finite-state strategies exist.
 
Indeed, in his master's thesis~\cite{Salzmann15}, Salzmann recently introduced the first notion of finite-state strategies in delay games and, using these reductions, presented an algorithm computing them for several types of acceptance conditions, e.g., parity conditions and related $\omega$-regular ones. However, the exact nature of finite-state strategies in delay games is not as canonical as for Gale-Stewart games. We discuss this issue in-depth in Sections~\ref{sec-fsindg} and \ref{sec-disc} by proposing two notions of finite-state strategies, a delay-oblivious one which yields large strategies in the size of the lookahead, and a delay-aware one that follows naturally from the reductions to Gale-Stewart games mentioned earlier. In particular, the number of states of the delay-aware strategies is independent of the size of the lookahead, but often larger in the size of the automaton recognizing the winning condition. However, this is offset by the fact that strategies of the second type are simpler to compute than the delay-oblivious ones and have overall fewer states, if the lookahead is large.  In comparison to Salzmann's notion, where strategies syntactically depend on a given automaton representing the winning condition, our strategies are independent of the representation of the winning condition and therefore more general. Also, our framework is more abstract and therefore applicable to a wider range of acceptance conditions (e.g., qualitative ones) and yields in general smaller strategies, but there are of course some similarities, which we discuss in detail.

To present these notions, we first introduce some definitions in Section~\ref{sec-prel}, e.g., delay games and finite-state strategies for Gale-Stewart games. After introducing the two notions of finite-state strategies for delay games in Section~\ref{sec-fsindg}, we show how to compute such strategies in Section~\ref{sec-construction}. To this end, we present a generic account of the reduction from delay games to Gale-Stewart games which subsumes, to the best of our knowledge, all decidability results presented in the literature. Furthermore, we show how to obtain the desired strategies from our construction. Then, in Section~\ref{sec-disc}, we compare the two different definitions of finite-state strategies for delay games proposed here and discuss their advantages and disadvantages. Also, we compare our approach to that of Salzmann. We conclude by mentioning some directions for further research in Section~\ref{sec-conc}.

Proofs and constructions omitted due to space restrictions can be found in the full version~\cite{Zimmermann17c}.

\paragraph*{Related Work}

As mentioned earlier, the existence of finite-state strategies is the technical core of many applications of infinite games, e.g., in reactive synthesis one synthesizes a correct-by-construction system from a given specification by casting the problem as an infinite game between a player representing the system and one representing the antagonistic environment. It is a winning strategy for the system player that yields the desired implementation, which is finite if the winning strategy is finite-state. Similarly, Gurevich and Harrington's game-based proof of Rabin's decidability theorem for monadic second-order logic over infinite binary trees~\cite{Rabin1969} relies on the existence of finite-state strategies.\footnote{The proof is actually based on positional strategies, a further restriction of finite-state strategies for arena-based games, because they are simpler to handle. Nevertheless, the same proof also works for finite-state strategies.}

These facts explain the need for studying the existence and properties of finite-state strategies in infinite games~\cite{Khoussainov03,Rabinovich09,LeRouxPauly16,Thomas94}. In particular, the seminal work by Dziembowski, Jurdzi\'{n}ski, and Walukiewicz~\cite{DziembowskiJW97} addressed the problem of determining upper and lower bounds on the size of finite-state winning strategies in games with Muller winning conditions. Nowadays, one of the most basic questions about a given winning condition is that about such upper and lower bounds. For most conditions in the literature, tight bounds are known, see, e.g.,~\cite{ChatterjeeHenzingerHorn11,Horn05,WallmeierHuettenThomas03}. But there are also surprising exceptions to that rule, e.g., generalized reachability games~\cite{FijalkowH13}. More recently, Colcombet, Fijalkow, and Horn presented a very general technique that yields tight upper and lower bounds on memory requirements in safety games, which even hold for games in  infinite arenas, provided their degree is finite~\cite{ColcombetFH14}.

%% file: content/prelims.tex
We denote the non-negative integers by~$\nats$. Given two $\omega$-words~$\alpha \in (\Sigma_0)^\omega$ and $\beta \in (\Sigma_1)^\omega$, we define ${ \alpha \choose \beta} = {\alpha(0) \choose \beta(0)} {\alpha(1) \choose \beta(1)} {\alpha(2) \choose \beta(2)} \cdots \in (\Sigma_0 \times \Sigma_1)^\omega$. Similarly, we define ${ x \choose y }$ for finite words~$x$ and $y$ with $\size{x} = \size{y}$.

\paragraph{\boldmath$\omega$-automata}

A (deterministic and complete) $\omega$-automaton is a tuple~$\aut = (Q, \Sigma, q_\initmark, \delta, \acc)$ where $Q$ is a finite set of states, $\Sigma$ is an alphabet, $q_\initmark \in Q$ is the initial state, $\delta \colon Q \times \Sigma \rightarrow Q$ is the transition function, and $\acc \subseteq \delta^\omega$ is the set of accepting runs (here, and whenever convenient, we treat $\delta$ as a relation~$\delta \subseteq Q \times \Sigma \times Q$).
A finite run~$\pi$ of $\aut$ is a sequence
$
\pi = (q_0, a_0, q_1)(q_1, a_1, q_2) \cdots (q_{i-2}, a_{i-2}, q_{i-1})(q_{i-1}, a_{i-1}, q_i) \in \delta^+$.
As usual, we say that $\pi$ starts in $q_0$, ends in $q_i$, and processes~$a_0\cdots a_{i-1} \in \Sigma^+$. Infinite runs on infinite words are defined analogously. If we speak of \emph{the} run of $\aut$ on $\alpha \in \Sigma^\omega$, then we mean the unique run of $\aut$ starting in $q_\initmark$ processing $\alpha$. The language~$L(\aut) \subseteq \Sigma^\omega$ of $\aut$ contains all those $\omega$-words whose run of $\aut$ is accepting. The size of $\aut$ is defined as $\size{\aut}=\size{Q}$.

This definition is very broad, which allows us to formulate our theorems as general as possible. In examples, we consider parity and Muller automata whose set of accepting runs is finitely represented: An $\omega$-automaton~$\aut = ( Q, \Sigma, q_\initmark, \delta, \acc )$ is a parity automaton, if 
$\acc = \set{(q_0, a_0, q_1) (q_1, a_1, q_2) (q_2, a_2, q_3) \cdots \in \delta^\omega \mid \text{$\limsup\nolimits_{i \rightarrow \infty} \col(q_i)  $ is even}}$ for some coloring~$\col \colon Q \rightarrow \nats$. To simplify our notation, define $\col(q,a,q') = \col(q)$. Furthermore, $\aut$ is a Muller automaton, if there is a family~$\curlyF \subseteq \pow{Q}$ of sets of states such that $\acc = \set{ \rho \in \delta^\omega \mid \infi(\rho) \in \curlyF}$, where $\infi(\rho)$ is the set of states visited infinitely often by $\rho$. 

\paragraph{Delay Games}

A delay function is a mapping~$f \colon \nats \rightarrow \nats\setminus \set{0}$, which is said to be constant if $f(i) =1$ for all $i>0$. A delay game~$\delaygame{L}$ consists of a delay function~$f$ and a winning condition~$L \subseteq (\SigmaI \times \SigmaO)^\omega$ for some alphabets~$\SigmaI$ and $\SigmaO$. Such a game is played in rounds~$i = 0,1,2, \ldots$ as follows: in round~$i$, first Player~$I$ picks a word~$x_i \in \SigmaI^{f(i)}$, then Player~$O$ picks a letter~$y_i \in \SigmaO$. Player~$O$ wins a play~$(x_0, y_0)(x_1, y_1)(x_2, y_2) \cdots $ if the outcome~${ x_0 x_1 x_2 \cdots \choose y_0 y_1 y_2 \cdots }$ is in $L$; otherwise, Player~$I$ wins.

A strategy for Player~$I$ in $\delaygame{L}$ is a mapping $\stratI \colon \SigmaO^* \rightarrow \SigmaI^*$ satisfying $\size{\stratI(w)} = f(\size{w})$ while a strategy for Player~$O$ is a mapping~$\stratO \colon \SigmaI^+ \rightarrow \SigmaO$. A play~$(x_0, y_0)(x_1, y_1)(x_2, y_2) \cdots $ is consistent with $\stratI$ if $x_i = \stratI(y_0 \cdots y_{i-1})$ for all $i$, and it is consistent with $\stratO$ if $y_i = \stratO(x_0 \cdots x_i)$ for all $i$. A strategy for Player~$\p \in \set{I,O}$ is winning, if every play that is consistent with the strategy is won by Player~$\p$.

An important special case are delay-free games, i.e., those with respect to the delay function~$f$ mapping every $i$ to $1$. In this case, we drop the subscript~$f$ and write $\game(L)$ for the game with winning condition~$L$. Such games are typically called Gale-Stewart games~\cite{GaleStewart53}.

\paragraph{Finite-state Strategies in Gale-Stewart Games}
\label{subsec-finitestate4galestewart}
A strategy for Player~$O$ in a Gale-Stewart game is still a mapping~$\stratO \colon \SigmaI^+ \rightarrow \SigmaO$. Such a strategy is said to be finite-state, if there is a deterministic finite transducer~$\strataut$ that implements $\stratO$ in the following sense: $\strataut$ is a tuple~$(Q, \SigmaI, q_\initmark, \delta, \SigmaO, \lambda)$ where $Q$ is a finite set of states, $\SigmaI$ is the input alphabet, $q_\initmark \in Q$ is the initial state, $\delta \colon Q \times \SigmaI \rightarrow Q$ is the deterministic transition function, $\SigmaO$ is the output alphabet, and $\lambda \colon Q \rightarrow \SigmaO$ is the output function. Let $\delta^*(x)$ denote the unique state that is reached by $\strataut$ when processing $x \in \SigmaI^*$ from $q_\initmark$. Then, the strategy~$\strat_\strataut$ implemented by $\strataut$ is defined as $\strat_\strataut(x) = \lambda(\delta^*(x))$. We say that a strategy is finite-state, if it is implementable by some transducer. Slightly abusively, we identify finite-state strategies with transducers implementing them and talk about finite-state strategies with some number of states. Thus, we focus on the \emph{state complexity} (e.g., the number of memory states necessary to implement a strategy) and ignore the other components of a transducer (which are anyway of polynomial size in $\size{Q}$, if we assume $\SigmaI$ and $\SigmaO$ to be fixed).

%% file: content/finitestateindelaygame.tex
Before we answer this question, we first ask what properties a finite-state strategy should have, i.e., what makes finite-state strategies in Gale-Stewart games useful and desirable? A strategy~$\stratO \colon \SigmaI^+ \rightarrow \SigmaO$ is in general an infinite object and does not necessarily have a finite representation. Furthermore, to execute such a strategy, one needs to store the whole sequence of moves made by Player~$I$ thus far: Unbounded memory is needed to execute it. 

On the other hand, a finite-state strategy is finitely described by a transducer~$\strataut$ implementing it.  To execute it, one only needs to store a single state of $\strataut$ and needs to have access to the transition function~$\delta$ and the output function~$\lambda$ of $\strataut$. Assume the current state is $q$ at the beginning of some round~$i$ (initialized with $q_\initmark$ before round~$0$). Then, Player~$I$ makes his move by picking some~$a \in \SigmaI$, which is processed by updating the memory state to $q' = \delta(q, a)$. Then, $\strataut$ prescribes picking $\lambda(q') \in \SigmaO$ and round~$i$ is completed.
Thus, there are two aspects that make finite-state strategies desirable: (1) the next move depends only on a finite amount of information about the history of the play, i.e., a state of the automaton, which is (2) easily updated. In particular, the strategy is completely specified by the transition function and the output function.

Further, there is a generic framework to compute such strategies by reducing them to arena-based games~(see, e.g., \cite{GraedelThomasWilke02} for an introduction to such games). As an example, consider a game~$\game(L(\aut))$ where $\aut$ is a parity automaton with set~$Q$ of states and transition function~$\delta$. We describe the construction of an arena-based parity game contested between Player~$I$ and Player~$O$ whose solution allows  us to compute the desired strategies (formal details are presented in the full version~\cite{Zimmermann17c}). The positions of Player~$I$ are transitions of $\aut$ while those of Player~$O$ are pairs~$(q,a)$ where $q \in Q$ and where $a$ is an input letter. From a vertex~$(q,{a \choose b}, q')$ Player~$I$ can move to every vertex~$(q',a')$ for $a' \in \SigmaI$, from which Player~$O$ can move to every vertex~$(q', {a' \choose b'}, \delta(q',{a' \choose b'})$ for $b' \in \SigmaO$. Finally, Player~$O$ wins a play, if the run constructed during the infinite play is accepting. It is easy to see that the resulting game is a parity game with~$ \size{\delta}\cdot(\size{\SigmaI}+1)$ vertices, and has the same winner as $\game(L(\aut))$. 
The winner of the arena-based game has a positional\footnote{A strategy in an arena-based games is positional, if its output only depends on the last vertex of the play's history, not on the full history (see, e.g.,~\cite{GraedelThomasWilke02}).} winning strategy~\cite{EmersonJutla91,Mostowski91}, which can be computed in quasipolynomial time~\cite{CJKLS16,FJSSW17,JL17}. Such a positional winning strategy can easily be turned into a finite-state winning strategy with $\size{Q} \cdot \size{\SigmaI}$ states for Player~$O$ in the game~$\game(L(\aut))$, which is implemented by an automaton with state set~$Q \times \SigmaI$. 
This reduction can be generalized to arbitrary classes of Gale-Stewart games whose winning condition is recognized by an $\omega$-automaton with set~$Q$ of states: if Player~$O$ has a finite-state strategy with $n$ states in the arena-based game obtained by the construction described above, then Player~$O$ has a finite-state winning strategy with $\size{Q} \cdot \size{\SigmaI}\cdot n$ states for the original Gale-Stewart game. Such a strategy is obtained by solving an arena-based game with $\size{\delta}\cdot(\size{\SigmaI}+1)$ vertices.

So, what is a finite-state strategy in a delay game? In the following, we discuss this question for the case of delay games with respect to constant delay functions, which is the most important case. In particular, constant lookahead suffices for all $\omega$-regular winning conditions~\cite{KleinZimmermann16}, i.e, Player~$O$ wins with respect to an arbitrary delay function if, and only if, she wins with respect to a constant one. Similarly, constant lookahead suffices for many quantitative conditions like (parameterized) temporal logics~\cite{KleinZimmermann16a} and parity conditions with costs~\cite{Zimmermann17}. For winning conditions given by parity automata, there is an exponential upper bound on the necessary constant lookahead. On the other hand, there are exponential lower bounds already for winning conditions specified by deterministic automata with reachability or safety acceptance (which are subsumed by parity acceptance).

\subsection{Delay-oblivious Finite-state Strategies for Delay Games}

Technically, a strategy for Player~$O$ in a delay game is still a mapping~$\stratO \colon \SigmaI^+ \rightarrow \SigmaO$. Hence, the definition of finite-state strategies via transducers as given above for Gale-Stewart games is also applicable to delay games. As a (cautionary) example, consider a delay game with winning condition~$\Leq = \set{ {\alpha \choose \alpha} \mid \alpha \in \set{0,1}^\omega }$, i.e., Player~$O$ just has to copy Player~$I$'s moves, which she can do with respect to every delay function: Player~$O$ wins $\delaygame{\Leq}$ for every $f$. However, a finite-state strategy has to remember the whole lookahead, i.e., those moves that Player~$I$ is ahead of Player~$O$, in order to copy his moves. Thus, an automaton implementing a winning strategy for Player~$O$ in $\delaygame{\Leq}$ needs at least~$\size{\set{0,1}}^d $ states, if $f$ is a constant delay function with $f(0) = d$. Thus, the memory requirements grow with the size of the lookahead granted to Player~$O$, i.e., lookahead is a burden, not an advantage. She even needs unbounded memory in the case of unbounded lookahead.

On the other hand, an advantage of this \myquot{delay-oblivious} definition is that finite-state strategies can be obtained by a trivial extension of the reduction presented for Gale-Stewart games above: now, states of Player~$I$ are from~$\delta \times \SigmaI^{d-1}$ and those of Player~$O$ are from $Q \times \SigmaI^{d}$. Player~$I$ can move from $((q,{a \choose b},q'),w)$ to $(q', wa')$ for $a' \in \SigmaI$ while Player~$O$ can move from $(q, aw)$ to $((q,{a \choose b}, \delta(q, {a \choose b})),w)$ for $b \in \SigmaO$. Intuitively, a state now additionally stores a queue of length~$d-1$, which contains the lookahead granted to Player~$O$. Coming back to the parity example, this approach yields a finite-state strategy with $\size{Q}\cdot\size{\SigmaI}^{d}$ states. To obtain such a strategy, one has to solve a parity game with $\size{\delta}\cdot(\size{\SigmaI}+1)\cdot\size{\SigmaI}^{d-1}$ vertices, which is of doubly-exponential size in $\size{\aut}$, if $d$ is close to the (tight) exponential upper bound. This can be done in doubly-exponential time, as it still has the same number of colors as the automaton~$\aut$. Again, this reduction can be generalized to arbitrary classes of delay games with constant delay whose winning conditions are recognized by an $\omega$-automaton with set~$Q$ of states: if Player~$O$ has a finite-state strategy with $n$ states in the arena-based game obtained by the construction, then Player~$O$ has a finite-state winning strategy with $\size{Q} \cdot \size{\SigmaI}^d\cdot n$ states for the delay game with constant lookahead of size~$d$. In general, $d$ factors exponentially into the size, as $n$ is the memory size required to win a game with $\bigo(\size{\SigmaI}^d)$ vertices. Also, to obtain the strategy for the delay game, one has to solve an arena-based game  with $\size{\delta}\cdot(\size{\SigmaI}+1)\cdot \size{\SigmaI}^{d-1}$~vertices.

\subsection{Block Games}

We show that one can do better than by decoupling the history tracking and the handling of the lookahead, i.e., by using \emph{delay-aware} finite-state strategies. In the delay-oblivious definition, we hardcode a queue into the arena-based game, which results in a blowup of the arena and therefore also in a blowup in the solution complexity and in the number of memory states for the arena-based game, which is turned into one for the delay game. To overcome this, we introduce a slight variation of delay games with respect to constant delay functions, so-called block games\footnote{Holtmann, Kaiser, and Thomas already introduced a notion of block game in connection to delay games~\cite{HoltmannKaiserThomas12}. However, their notion differs from ours in several aspects. Most importantly, in their definition, Player~$I$ determines the length of the blocks (within some bounds specified by $f$) while our block length is fixed.}, present a notion of finite-state strategy in block games, and show how to transfer strategies between delay games and block games. Then, we show how to solve block games and how to obtain finite-state strategies for them.

 The motivation for introducing block games is to eliminate the queue containing the letters Player~$I$ is ahead of Player~$O$, which is cumbersome to maintain, and causes the blowup in the case of games with winning condition~$\Leq$. Instead, in a block game, both players pick blocks of letters of a fixed length with Player~$I$ being one block ahead to account for the delay, i.e., Player~$I$ has to pick two blocks in round~$0$ and then one in every round, as does Player~$O$ in every round. This variant of delay games lies implicitly or explicitly at the foundation of all arguments establishing upper bounds on the necessary lookahead and at the foundation of all algorithms solving delay games~\cite{HoltmannKaiserThomas12, KleinZimmermann16,KleinZimmermann16a,Zimmermann16,Zimmermann17}. Furthermore, we show how to transform a (winning) strategy for a delay game into a (winning) strategy for a block game and vice versa, i.e., Player~$O$ wins the delay game if, and only if, she wins the corresponding block game.\footnote{Due to their prevalence and importance for solving delay games, one could even argue that the notion of block games is more suitable to model delay in infinite games.}

Formally, the block game~$\blockgame{d}{L}$, where $d \in \nats\setminus\set{0}$ is the block length and where $L \subseteq (\SigmaI\times\SigmaO)^\omega$ is the winning condition, is played in rounds as follows: in round~$0$, Player~$I$ picks two blocks~$\block{a_0},\block{a_1} \in \SigmaI^d$, then Player~$O$ picks a block~$\block{b_0}\in\SigmaO^d$. In round~$i>0$, Player~$I$ picks a block~$\block{a_{i+1}} \in \SigmaI^d$, then Player~$O$ picks a block~$\block{b_i}\in\SigmaO^d$. Player~$O$ wins the resulting play~$\block{a_0} \block{a_1} \block{b_0} \block{a_2} \block{b_1} \cdots$, if the outcome~${\block{a_0}\block{a_1}\block{a_2} \cdots  \choose \block{b_0}\block{b_1}\block{b_2} \cdots}$ is in $L$. 

A strategy for Player~$I$ in $\blockgame{d}{L}$ is a map~$\stratI \colon (\SigmaO^{d})^* \rightarrow (\SigmaI^{d})^2 \cup \SigmaI^{d}$ such that $\stratI(\epsilon) \in (\SigmaI^{d})^2$ and $\stratI(\block{b_0} \cdots \block{b_i}) \in \SigmaI^{d}$ for $i\ge 0$. A strategy for Player~$O$ is a map~$\stratO \colon (\SigmaI^d)^* \rightarrow \SigmaO^d$. A play~$\block{a_0} \block{a_1} \block{b_0} \block{a_2} \block{b_1} \cdots$ is consistent with $\stratI$, if $(\block{a_0}, \block{a_1}) = \stratI(\epsilon)$ and $\block{a_i} = \stratI(\block{b_0} \cdots \block{b_{i-2}})$ for every $i \ge 2$; it is consistent with $\stratO$ if $\block{b_i} = \stratO(\block{a_0} \cdots \block{a_{i+1}})$ for every $i \ge 0$. Winning strategies and winning a block game are defined as for delay games.

In the following, we call strategies for block games delay-aware and strategies for delay games delay-oblivious. The next lemma relates delay games with constant lookahead and block games: for a given winning condition, Player~$O$ wins a delay game with winning condition~$L$ (with respect to some delay function) if, and only if, she wins a block game with winning condition~$L$ (for some block size). 

\begin{lemma}
\label{lemma-delaygamesvsblockgames}
Let $L \subseteq (\SigmaI \times \SigmaO)^\omega$.
\begin{enumerate}
	\item\label{lemma-delaygamesvsblockgames-delay2block}
	 If Player~$O$ wins $\delaygame{L}$ for some constant delay function~$f$, then she also wins $\blockgame{f(0)}{L}$.
	\item\label{lemma-delaygamesvsblockgames-block2delay}
	 If Player~$O$ wins $\blockgame{d}{L}$, then she also wins $\delaygame{L}$ for the constant delay function~$f$ with $f(0) = 2d$. 
\end{enumerate}
\end{lemma}

\begin{proof}
\ref{lemma-delaygamesvsblockgames-delay2block}.) 
Let $\stratO \colon \SigmaI^+ \rightarrow \SigmaO$ be a winning strategy for Player~$O$ in $\delaygame{L}$ and fix $d = f(0)$. Now, define $\stratO' \colon (\SigmaI^d)^* \rightarrow (\SigmaO)^d$ for Player~$O$ in $\blockgame{d}{L}$ via $\stratO'(\block{a_0} \cdots \block{a_i} \block{a_{i+1}}) = \beta(0) \cdots \beta(d-1)$ with
$\beta(j) = \stratO(\block{a_0} \cdots \block{a_{i}} \alpha(0) \alpha(1) \cdots \alpha(j-1))$
for $\block{a_{i+1}} = \alpha(0) \alpha(1) \cdots \alpha(d-1)$.

A straightforward induction shows that for every play consistent with $\stratO'$ there is a play consistent with $\stratO$ that has the same outcome. Thus, as $\stratO$ is a winning strategy, so is $\stratO'$.

\ref{lemma-delaygamesvsblockgames-block2delay}.) Now, let $\stratO' \colon (\SigmaI^d)^* \rightarrow (\SigmaO)^d$ be a winning strategy for Player~$O$ in $\blockgame{d}{L}$. We define $\stratO \colon \SigmaI^+ \rightarrow \SigmaO$ for Player~$O$ in $\delaygame{L}$. To this end, let $x \in \SigmaI^+$ be a possible input occurring during a play. Hence, by the choice of $f$, we obtain $\size{x} \ge f(0) = 2d$. Thus, we can decompose $x$ into $x = \block{a_0} \cdots \block{a_{i}} x'$ such that $i \ge 1$, each $\block{a_{i'}}$ is a block over $\SigmaI$ of length~$d$ and $\size{x'} < d$. 
Now, let $\stratO'(\block{a_0} \cdots \block{a_i}) = \beta(0) \cdots \beta(d-1)$. Then, we define $\stratO(x) = \beta(\size{x'})$.

Again, a straightforward induction shows that for every play consistent with $\stratO$ there is a play consistent with $\stratO'$ that has the same outcome. Thus, $\stratO$ is a winning strategy.
\end{proof}

\subsection{Delay-aware Finite-state Strategies in Block Games}

Now fix a block game~$\blockgame{d}{L}$ with $L \subseteq (\SigmaI \times \SigmaO)^\omega$. A finite-state strategy for Player~$O$ in $\blockgame{d}{L}$ is implemented by a transducer~$\strataut = (Q, \SigmaI, q_\initmark, \delta, \SigmaO, \lambda)$ where $Q$, $\SigmaI$, and $q_\initmark$ are defined as in Subsection~\ref{sec-prel}.\ref{subsec-finitestate4galestewart}. However, the transition function~$\delta\colon Q \times \SigmaI^d \rightarrow Q$ processes full input blocks and the output function~$\lambda \colon Q \times \SigmaI^d \times \SigmaI^d \rightarrow \SigmaO^d$ maps a state and a pair of input blocks to an output block. The strategy~$\strat_\strataut$ implemented by $\strataut$ is defined as $\strat_\strataut(\block{a_0} \cdots \block{a_i}) = \lambda(\delta^*(\block{a_0} \cdots \block{a_{i-2}}), \block{a_{i-1}}, \block{a_i} )$ for $i \ge 1$. 

Again, we identify delay-aware strategies with transducers implementing them and are interested in the number of states of the transducer. This definition captures the amount of information that is differentiated in order to implement the strategy. Note however, that it ignores the  representation of the transition and the output function. These are no longer \myquot{small} (in $\size{Q}$), as it is the case for transducers implementing strategies for Gale-Stewart games. When focussing on executing such strategies, these factors become relevant, but for our purposes they are not: 
We have  decoupled the history tracking from the lookahead-handling. The former is implemented by the automaton as usual while the latter is taken care of by the output function. In particular, the size of the automaton is (a-priori) independent of the block size. In the conclusion, we revisit the issue of presenting the transition and the output function. 

In the next section, we present a very general approach to computing finite-state strategies for block games whose winning conditions are specified by automata with acceptance conditions that satisfy a certain aggregation property. For example, for block games with winning conditions given by deterministic parity automata, we obtain a strategy implemented by a transducer with exponentially many states, which can be obtained by solving a parity game of exponential size. In both aspects, this is an exponential improvement over the delay-oblivious variant for classical delay games.

To conclude the introduction of block games, we strengthen Lemma~\ref{lemma-delaygamesvsblockgames} to transfer finite-state strategies between delay games and block games.

\begin{lemma}\label{lemma-delaygamesvsblockgames-fs}
Let $L \subseteq (\SigmaI \times \SigmaO)^\omega$.
\begin{enumerate}
	\item\label{lemma-delaygamesvsblockgames-fs_delay2block}
	 If Player~$O$ has a delay-oblivious finite-state winning strategy for $\delaygame{L}$ with $n$ states for some constant delay function~$f$, then she also has a delay-aware finite-state winning strategy for $\blockgame{f(0)}{L}$ with $n$ states.
	\item\label{lemma-delaygamesvsblockgames-fs_block2delay}
	 If Player~$O$ has a delay-aware finite-state winning strategy for $\blockgame{d}{L}$ with $n$ states, then she also has a delay-oblivious finite-state winning strategy for $\delaygame{L}$  with $n \cdot \size{\SigmaI}^{2d} $ states for the constant delay function~$f$ with $f(0) = 2d$. 
\end{enumerate}
\end{lemma}

\begin{proof}
It is straightforward to achieve the strategy transformations described in the proof of Lemma~\ref{lemma-delaygamesvsblockgames} by transforming transducers that implement finite-state strategies. 
\end{proof}

The blowup in the direction from block games to delay games is in general unavoidable, as finite-state winning strategies for the game~$\delaygame{\Leq}$ need at least $2^d$ states to store the lookahead while winning strategies for the block game need only one state, independently of the block size.

%% file: content/construction.tex
The aim of this section is twofold. Our main aim is to compute finite-state strategies for block games (and, by extension, for delay games with constant lookahead). We do so by presenting a general framework for analyzing delay games with winning conditions specified by $\omega$-automata whose acceptance conditions satisfy a certain aggregation property. The technical core is a reduction to a Gale-Stewart game, i.e., we remove the delay from the game. This framework yields upper bounds on the necessary (constant) lookahead to win a given game, but also allows us to determine the winner and a finite-state winning strategy, if the resulting Gale-Stewart game can be effectively solved. 

Slightly more formally, let $\aut$ be the automaton recognizing the winning condition of the block game. Then, the winning condition of the Gale-Stewart game constructed in the reduction is recognized by an automaton~$\autb$ that can be derived from $\aut$. In particular, the acceptance condition of $\autb$ simulates the acceptance condition of $\aut$. Many types of acceptance conditions are preserved by the simulation, e.g., starting with a parity automaton~$\aut$, we end up with a parity automaton~$\autb$. Thus, the resulting Gale-Stewart game can be effectively solved.

Our second aim is to present a framework as general as possible to obtain upper bounds on the necessary lookahead and on the solution complexity for a wide range of winning conditions. In fact, our framework is a generalization and abstraction of techniques first developed for the case of $\omega$-regular winning conditions~\cite{KleinZimmermann16}, which were later generalized to other winning conditions~\cite{KleinZimmermann16a,Zimmermann16,Zimmermann17}. Here, we cover all these results in a uniform way.

\subsection{Aggregations}

Let us begin by giving some intuition for the construction. The winning condition of the game is recognized by an automaton~$\aut$. Thus, as usual, the exact input can be abstracted away, only the induced behavior in $\aut$ is relevant. Such a behavior is characterized by the state transformations induced by processing the input and by the effect on the acceptance condition triggered by processing it. For many acceptance conditions, this effect can be aggregated, e.g., for parity conditions, one can decompose runs into non-empty pieces and then only consider the maximal colors of the pieces. For quantitative winning conditions, one typically needs an additional bound on the lengths of these pieces (cp.~\cite{Zimmermann16,Zimmermann17}).

Thus, we begin by introducing two types of such aggregations of different strength. Fix an $\omega$-automaton~$\aut = (Q, \Sigma, q_\initmark, \delta, \acc )$ and let $s \colon \delta^+ \rightarrow M$ for some finite set~$M$. Given a decomposition~$(\pi_i)_{i\in\nats}$ of a run $\pi_0 \pi_1 \pi_2 \cdots$ into non-empty pieces~$\pi_i\in \delta^+$ we define $s((\pi_i)_{i\in\nats}) = s(\pi_0)s(\pi_1)s(\pi_2)\cdots \in M^\omega$.

\begin{itemize}
	
	\item We say that $s$ is a strong aggregation (function) for $\aut$, if for all decompositions~$(\pi_i)_{i\in\nats}$ and $(\pi_i')_{i\in\nats}$ of runs~$\rho = \pi_0 \pi_1 \pi_2 \cdots $ and $\rho' = \pi_0' \pi_1' \pi_2' \cdots $ with $\sup_i \size{\pi_i'} < \infty$ and $s((\pi_i)_{i\in\nats}) = s((\pi_i')_{i\in\nats})$: $\rho \in \acc \Rightarrow \rho'\in \acc$.

	\item We say that $s$ is a weak  aggregation (function) for $\aut$, if for all decompositions~$(\pi_i)_{i\in\nats}$ and $(\pi_i')_{i\in\nats}$  of runs~$\rho = \pi_0 \pi_1 \pi_2 \cdots $ and $\rho' = \pi_0' \pi_1' \pi_2' \cdots $ with $\sup_i \size{\pi_i} < \infty$, $\sup_i \size{\pi_i'} < \infty$, and $s((\pi_i)_{i\in\nats}) = s((\pi_i')_{i\in\nats})$: $\rho \in \acc \Rightarrow \rho'\in \acc$.

\end{itemize}

\begin{example}\hfill
\label{example-aggregation}

\begin{itemize}
	
	\item The function~$s_\parity \colon \delta^+ \rightarrow \col(Q)$ defined as $s_\parity(t_0 \cdots t_i) = \max_{0\le j \le i}\col(t_j)$ is a strong aggregation for a parity automaton~$(Q, \Sigma, q_\initmark, \delta, \acc)$ with coloring~$\col$ (recall that $\col(q,a,q') =\col(q)$).
	
	\item The function~$s_\muller \colon \delta^+ \rightarrow \pow{Q}$ defined as $s_\muller((q_0, a_0, q_1) \cdots (q_{n}, a_n, q_{n+1})) = \set{q_0, q_1, \ldots, q_{n}}$ is a strong aggregation for a Muller automaton~$(Q, \Sigma, q_\initmark, \delta, \acc)$. 
	
	\item The exponential time algorithm for delay games with winning conditions given by parity automata with costs, a quantitative generalization of parity automata, is based on a strong aggregation~\cite{Zimmermann17}. 
	
	\item The algorithm for delay games with winning conditions given by max automata~\cite{Bojanczyk11}, another quantitative automaton model, is based on a weak aggregation~\cite{Zimmermann16}.
	
\end{itemize}
\end{example}

Due to symmetry, we can replace the implication~$\rho \in \acc \Rightarrow \rho' \in \acc$ by an equivalence in the definition of a weak aggregation. Also, every strong aggregation is trivially a weak one as well. 

Let us briefly comment on the difference between strong and weak aggregations using the examples of parity automata with costs and max-automata: the acceptance condition of the former automata is a boundedness condition on some counters while the acceptance condition of the latter is a boolean combination of boundedness and unboundedness conditions on some counters. The aggregations for these acceptance conditions capture whether a piece of a run induces an increment of a counter or not, but abstract away the actual number of increments if it is non-zero. 
Now, consider the parity condition with costs, which requires to bound the counters. Assume the counters in some run~$\pi_0 \pi_1 \pi_2 \cdots$ are bounded and that we have pieces~$\pi_i'$ of bounded length having the same aggregation. Then, the increments in some piece~$\pi_i'$ have at least one corresponding increment in $\pi_i$. Thus, if a counter in $\pi_0' \pi_1' \pi_2' \cdots$ is unbounded, then it is also unbounded in $\pi_0 \pi_1 \pi_2 \cdots$, which yields a contradiction. Hence, the implication~$\pi_0 \pi_1 \pi_2 \cdots \in \acc \Rightarrow \pi_0' \pi_1' \pi_2' \cdots\in \acc$ holds. For details, see~\cite{Zimmermann17}.
On the other hand, to preserve boundedness and unboundedness properties, one needs to bound the length of the $\pi_i'$ and the length of the $\pi_i$. Hence, there is only a weak aggregation for max-automata. Again, see~\cite{Zimmermann16} for details.

Given a weak  aggregation~$s$ for $\aut$ with acceptance condition~$\acc$, let 
\[
s(\acc) = \set{
s((\pi_i)_{i\in\nats}) \mid \pi_0 \pi_1 \pi_2 \cdots \in \acc \text{ is an accepting run of }\aut \text{ with } \sup\nolimits_i \size{\pi_i} < \infty}
.\]

Next, we consider aggregations that are trackable by automata. A monitor for an automaton~$\aut$ with transition function~$\delta$ is a tuple~$\mon = (M,\bot, \update)$ where $M$ is a finite set of memory elements, $\bot \notin M$ is the empty memory element, and $\update \colon M_\bot \times \delta \rightarrow M$ is an update function, where we use $M_\bot = M \cup \set{\bot}$. Note that the empty memory element~$\bot$ is only used to initialize the memory, it is not in the image of $\update$. 
We say that $\mon$ computes the function~$s_\mon \colon \delta^+ \rightarrow M$ defined by $s_\mon(t) = \update(\bot, t)$ and $s_\mon(\pi\cdot t) = \update(s_\mon(\pi), t)$ for $\pi \in \delta^+$ and $t \in \delta$. 

\begin{example}
\label{example-monitor}
Recall  Example~\ref{example-aggregation}. The strong aggregation~$s_\parity$ for a parity automaton is computed by the monitor~$(\col(Q), \bot, (c,t) \mapsto \max\set{c,\col(t)})$, where $\bot < c$ for every $c \in \col(Q)$.
\end{example}

Next, we take the product of $\aut$ and the monitor~$\mon$ for $\aut$, which simulates $\aut$ and simultaneously aggregates the acceptance condition. Formally, we define the product as $\aut \times \mon = (Q \times M_\bot, (q_\initmark, \bot), \Sigma, \delta', \emptyset)$ where $\delta'((q, m), a) = (q', \update(m, (q, a, q'))) $ for $q' = \delta(q, a)$. Note that $\aut\times\mon$ has an empty set of accepting runs, as these are irrelevant to us.

\subsection{Removing Delay via Aggregation}

Consider a play prefix in a delay game~$\delaygame{L(\aut)}$: Player~$I$ has produced a sequence~$\alpha(0) \cdots \alpha(i)$ of letters while Player~$O$ has produced $\beta(0) \cdots \beta(i')$ with, in general, $i'<i$. Now, she has to determine $\beta(i'+1)$. The automaton~$\aut \times \mon$ can process the joint sequence~${\alpha(0) \cdots \alpha(i') \choose \beta(0)\cdots \beta(i')}  $, but not the sequence~$\alpha(i'+1) \cdots \alpha(i)$, as Player~$O$ has not yet picked the letters~$\beta(i'+1) \cdots \beta(i)$. However, one can determine which states are reachable by some completion~${\alpha(i'+1) \cdots \alpha(i) \choose \beta(i'+1) \cdots \beta(i)}$ by  projecting away $\SigmaO$ from $\aut\times\mon$.

Thus, from now on assume $\Sigma = \SigmaI \times \SigmaO$ and define $\delta_P \colon 2^{Q\times M_\bot} \times \SigmaI \rightarrow \pow{Q \times M}$ via 
\[
\delta_P (S, a) = \left\{\left.\delta'\left( (q,m), {a \choose b} \right) \right| (q,m) \in S \text{ and } b \in \SigmaO\right\}.\]
Intuitively, $\delta_P$ is obtained as follows: take $\aut\times\mon$, project away $\SigmaO$, and apply the power set construction (while discarding the anyway empty acceptance condition). Then, $\delta_P$ is the transition function of the resulting deterministic automaton.
As usual, we extend $\delta_P$ to $\delta_P^+ \colon 2^{Q \times M_\bot} \times \SigmaI^+ \rightarrow 2^{Q\times M}$ via $\delta_P^+(S,a) = \delta_P(S,a)$ and $\delta_P^+(S, wa) = \delta_P( \delta_P^+(S, w), a)$. 

\begin{numberedremark}
\label{remark-powersetcharac}
The following are equivalent for $q \in Q$ and $w \in \SigmaI^+$:
\begin{enumerate}
	\item $(q',m') \in \delta_P^+(\set{(q, \bot)},w)$. 
	\item There is a $w' \in (\SigmaI \times \SigmaO)^+ $ whose projection to $\SigmaI$ is $w$ such that the run~$\pi$ of $\aut$ processing $w'$ starting from $q$ ends in $q'$ and satisfies $s_\mon(\pi) = m'$.
\end{enumerate}
\end{numberedremark}

 We use this property to define an equivalence relation formalizing the idea that words having the same behavior in $\aut \times \mon$ do not need to be distinguished. To this end, to every $w \in \SigmaI^+$ we assign the transition summary~$r_w \colon Q \rightarrow \pow{Q \times M}$ defined via $
r_w(q) = \delta_P^+(\set{(q, \bot)}, w)$.
Having the same transition summary is a finite equivalence relation~$\equiv$ over $\SigmaI^+$ whose index is bounded by $2^{\size{Q}^2\cdot\size{M}}$. For an $\equiv$-class~$S = \equivclass{w}$ define $r_S = r_w$, which is independent of representatives. Let $\R$ be the set of infinite $\equiv$-classes. 

Now, we define a Gale-Stewart game in which Player~$I$ determines an infinite sequence of equivalence classes from $\R$. By picking representatives, this induces a word~$\alpha \in \SigmaI^\omega$. Player~$O$ picks states~$(q_i,m_i)$ such that the $m_i$ aggregate a run of $\aut$ on some completion~${\alpha \choose \beta}$ of $\alpha$. Player~$O$ wins if the $m_i$ imply that the run of $\aut$ on ${\alpha \choose \beta}$ is accepting. To account for the delay, Player~$I$ is always one move ahead, which is achieved by adding a dummy move for Player~$O$ in round~$0$.

Formally, in round~$0$, Player~$I$ picks an $\equiv$-class~$S_0 \in \R$ and Player~$O$ has to pick $(q_0, m_0) = (q_\initmark, \bot)$. In round~$i>0$, first Player~$I$ picks an $\equiv$-class~$S_i \in \R$, then Player~$O$ picks a state~$(q_i, m_i) \in r_{S_{i-1}}(q_{i-1})$ of the product automaton. Player~$O$ wins the resulting play~$S_0 (q_0, m_0) S_1 (q_1, m_1) S_2 (q_2, m_2) \cdots$ if $m_1 m_2 m_3 \cdots \in s_\mon(\acc)$ (note that $m_0$ is ignored). The notions of (finite-state and winning) strategies are inherited from Gale-Stewart games, as this game is indeed such a game~$\game(L(\autb))$ for some automaton~$\autb$ of size~$\size{\R} \cdot \size{Q}\cdot\size{M} $ which can be derived from $\aut$ and $\mon$.

Formally, we define $\autb = (\R \times Q \times M_\bot, \R \times (Q \times M), (S_\initmark, q_\initmark, m_\initmark), \delta', \acc')$ for some arbitrary~$S_\initmark \in \R$, some arbitrary~$m_\initmark \in M$, $\delta'((S,q,m),{S' \choose (q',m')}) = (S',q',m')$, and 
$(S_0,q_0,m_0) (S_1,q_1,m_1) (S_2,q_2,m_2) \cdots \in \acc'$ if, and only if, 
\begin{itemize}
	\item $(q_0, m_0) = (q_\initmark,\bot)$,
	\item $(q_i,m_i) \in r_{S_{i-1}}(q_{i-1})$ for all $i >0$, and
	\item $m_1 m_2 m_3 \cdots \in s_\mon(\acc)$.
\end{itemize}
It is straightforward to prove that $\autb$ has the desired properties.

Note that, due to our very general definition of acceptance conditions, we are able to express the local consistency requirement~\myquot{$(q_i,m_i) \in r_{S_{i-1}}(q_{i-1})$} using the acceptance condition. For less general acceptance modes, e.g., parity, one has to check this property using the state space of the automaton, which leads to a polynomial blowup, as one has to store each~$S_{i-1}$ for one transition.

\begin{theorem}\label{thm-main}
Let $\aut$ be an $\omega$-automaton and let $\mon$ be a monitor for $\aut$ such that $s_\mon$ is a strong aggregation for $\aut$, let $\autb$ be constructed as above, and define $d = 2^{\size{Q}^2\cdot\size{M_\bot}}$.

\begin{enumerate}
	
	\item\label{thm-main-delay2delayfree} If Player~$O$ wins $\delaygame{L(\aut)}$ for some delay function~$f$, then she also wins $\game(L(\autb))$.
	
	\item\label{thm-main-delayfree2block} If Player~$O$ wins $\game(L(\autb))$, then she also wins the block-game~$\blockgame{d}{L(\aut)}$. Moreover, if she has a finite-state winning strategy for $\game(L(\autb))$ with $n$ states, then she has a delay-aware finite-state winning strategy for $\blockgame{d}{L(\aut)}$ with $n $ states.
	
\end{enumerate}
\end{theorem}

By applying both implications and Item~\ref{lemma-delaygamesvsblockgames-block2delay} of Lemma~\ref{lemma-delaygamesvsblockgames}, we obtain upper bounds on the complexity of determining for a given~$\aut$ whether Player~$O$ wins $\delaygame{L(\aut)}$ for some $f$ and on the necessary constant lookahead necessary to do so.

\begin{corollary}
Let $\aut$, $\autb$, and $d$ be as in Theorem~\ref{thm-main}. Then, the following are equivalent:
\begin{enumerate}
	\item Player~$O$ wins $\delaygame{L(\aut)}$ for some delay function~$f$.
	\item Player~$O$ wins $\delaygame{L(\aut)}$ for the constant delay function~$f$ with $f(0) = 2d$.
	\item Player~$O$ wins $\game(L(\autb))$.
\end{enumerate}
\end{corollary}

Thus, determining whether, given $\aut$, Player~$O$ wins $\delaygame{L(\aut)}$ for some $f$ is achieved by determining the winner of the Gale-Stewart game~$\game(L(\autb))$ and, independently, we obtain an exponential (in $\size{Q} \cdot \size{M}$) upper bound on the necessary constant lookahead.

\begin{example}
Continuing our example for the parity acceptance condition, we obtain the exponential upper bound~$2^{\size{Q}^2\cdot \size{\col(Q)}+2}$ on the constant lookahead necessary to win the delay game and an exponential-time algorithm for determining the winner, as $\autb$ has exponentially many states, but the same number of colors as $\aut$. Both upper bounds are tight~\cite{KleinZimmermann16}.
\end{example}

In case there is no strong aggregation for $\aut$, but only a weak  one, one can show that finite-state strategies exist, if Player~$O$ wins with respect to some constant delay function at all.

\begin{theorem}\label{thm-main2}

Let $\aut$ be an $\omega$-automaton and let $\mon$ be a monitor for $\aut$ such that $s_\mon$ is a weak aggregation for $\aut$, let $\autb$ be constructed as above, and define $d = 2^{\size{Q}^2\cdot\size{M_\bot}}$.
\begin{enumerate}
	
	\item\label{thm-main2-delay2delayfree} If Player~$O$ wins $\delaygame{L(\aut)}$ for some constant delay function~$f$, then she also wins $\game(L(\autb))$.
	
	\item\label{thm-main2-delayfree2block} If Player~$O$ wins $\game(L(\autb))$, then she also wins the block-game~$\blockgame{d}{L(\aut)}$. Moreover, if she has a finite-state winning strategy for $\game(L(\autb))$  with $n$ states, then she has a delay-aware finite-state winning strategy for $\blockgame{d}{L(\aut)}$ with $n $ states.
		
\end{enumerate}

\end{theorem}

Again, we obtain upper bounds on the solution complexity (here, with respect to constant delay functions) and on the necessary constant lookahead.

\begin{corollary}
Let $\aut$, $\autb$, and $d$ be as in Theorem~\ref{thm-main2}. Then, the following are equivalent:
\begin{enumerate}
	\item Player~$O$ wins $\delaygame{L(\aut)}$ for some constant delay function~$f$.
	\item Player~$O$ wins $\delaygame{L(\aut)}$ for the constant delay function~$f$ with $f(0) = 2d$.
	\item Player~$O$ wins $\game(L(\autb))$.
\end{enumerate}
\end{corollary}

%% file: content/disc.tex
Let us compare the two approaches presented in the previous section with three use cases: delay games whose winning conditions are given by deterministic parity automata, by deterministic Muller automata, and by LTL formulas. All formalisms only define $\omega$-regular languages, but vary in their succinctness. 

The following facts about arena-based games will be useful for the comparison:
\begin{itemize}
	\item The 
	winner of arena-based parity games 
 has positional winning strategies~\cite{EmersonJutla91,Mostowski91}, i.e., finite-state strategies with a single state.
	\item 
	The winner of an arena-based Muller game has a finite-state strategy with $n!$ states~\cite{McNaughton93}, where $n$ is the number of vertices of the arena.
	\item 
	The winner of an arena-based LTL game has a finite-state strategy with $2^{2^{\size{\phi}}}$ states~\cite{PnueliRosner89a}, where $\phi$ is the formula specifying the winning condition.
\end{itemize}
Also, we need the following bounds on the necessary lookahead in delay games: 
\begin{itemize}
	\item In delay games whose winning conditions are given by deterministic parity automata, exponential (in the size of the automata) constant lookahead is both sufficient and in general necessary~\cite{KleinZimmermann16}.
	\item In delay games whose winning conditions are given by deterministic Muller automata, doubly-exponential (in the size of the automata) constant lookahead is sufficient. This follows from the transformation of deterministic Muller automata into deterministic parity automata of exponential size (see, e.g.,~\cite{GraedelThomasWilke02}). However, the best lower bound is the exponential one for parity automata, which are also Muller automata.

 	\item In delay games whose winning conditions are given by LTL formulas, triply-exponential (in the size of the formula) constant lookahead is both sufficient and in general necessary~\cite{KleinZimmermann16a}.

\end{itemize}

Using these facts, we obtain the following complexity results for finite-state strategies: Figure~\ref{fig-naive} shows the upper bounds on the number of states of delay-oblivious finite-state strategies for delay games and on the number of states of delay-aware finite-state strategies for block games. In all three cases, the former strategies are at least exponentially larger. This illustrates the advantage of decoupling tracking the history from managing the lookahead. 

\begin{figure}[h]
\centering
\begin{tabular}{llll}
 &  parity  &  Muller & LTL \\
\toprule
\textbf{delay-oblivious}  & doubly-exp. & 
quadruply-exp.&
quadruply-exp. 
\\

\midrule
\textbf{delay-aware}   & exp.  & 
doubly-exp.&
triply-exp. 

\end{tabular}	
\caption{Memory size for delay-oblivious strategies (for delay games) and delay-aware finite-state strategies (for block games), measured in the size of the representation of the winning condition. For the sake of readability, we only present the orders of magnitude, but not exact values.}
\label{fig-naive}
\end{figure}


Finally, let us compare our approach to that of Salzmann. Fix a delay game~$\delaygame{L(\aut)}$ and assume Player~$I$ has picked $\alpha(0) \cdots \alpha(i)$ while Player~$O$ has picked $\beta(0) \cdots \beta(i')$ with $i'<i$. His strategies are similar to our delay-aware ones for block games. The main technical difference is that his strategies have access to the state reached by $\aut$ when processing ${\alpha(0) \cdots \alpha(i')  \choose \beta(0) \cdots \beta(i')}$. Thus, his strategies explicitly depend on the specification automaton~$\aut$ while ours are independent of it. In general, his strategies are therefore smaller than ours, as our transducers have to simulate~$\aut$ if they need access to the current state. On the other hand, our aggregation-based framework is more general and readily applicable to quantitative winning conditions as well, while he only presents results for selected qualitative conditions like parity, weak parity, and Muller.

%% file: content/conc.tex
We have presented a very general framework for analyzing delay games. If the automaton recognizing the winning condition satisfies a certain aggregation property, our framework yields upper bounds on the necessary lookahead to win the game, an algorithm for determining the winner (under some additional assumptions on the acceptance condition), and finite-state winning strategies for Player~$O$, if she wins the game at all. These results cover all previous results on the first two aspects (although not necessarily with optimal complexity of determining the winner). 

Thereby, we have lifted another important aspect of the theory of infinite games to the setting with delay. However, many challenging open questions remain, e.g., a systematic study of memory requirements in delay games is now possible. For delay-free games, tight upper and lower bounds on these requirements are known for almost all winning conditions.

Another exciting question concerns the tradeoff between memory and amount of lookahead: can one trade memory for lookahead? In other settings, such tradeoffs exist, e.g., lookahead allows Player~$O$ to improve the quality of her strategies~\cite{Zimmermann17}. Salzmann has presented some tradeoffs between memory and lookahead, e.g., linear lookahead allows exponential reductions in memory size in comparison to delay-free strategies~\cite{Salzmann15}. In current work, we investigate whether these results are inherent to his setting, which differs subtly from the one proposed here, or whether they exist in our setting as well. 

Finite-state strategies in arena-based games are typically computed by game reductions, which turn a game with a complex winning condition into one in a larger arena with a simpler winning condition. In future work, we plan to lift this approach to delay games. Note that the algorithm for computing finite-state strategies presented here can already be understood as a reduction, as we turn a delay game into a Gale-Stewart game. This removes the delay, but preserves the type of winning condition. However, it is also conceivable that staying in the realm of delay games yields better results, i.e., by keeping the delay while simplifying the winning condition. In future work, we address this question.

In our study here we focussed on the state complexity of the automata implementing the strategies, i.e., we measure the quality of a strategy in the number of states of a transducer implementing it. However, this is not the true size of such a machine, as we have ignored the need to represent the transition function and the output function, which have an exponential domain (in the block size) in the case of delay-aware strategies. Thus, when represented as lookup tables, they are prohibitively large. However, our delay-removing reduction hints at these functions also being implementable by transducers. For the transition function this is straightforward; in current work, we investigate whether this is also possible for the output function.

Finally, in future work we will determine the complexity of computing finite-state strategies in delay games and investigate notions of finite-state strategies for Player~$I$, which should be much simpler since he does not have to deal with the lookahead. 

\paragraph{Acknowledgements} The author is very grateful to the anonymous reviewers whose feedback significantly improved the exposition. 